\documentclass{article}

\usepackage[algo2e,linesnumbered,vlined,ruled]{algorithm2e}
\usepackage{cite}
\usepackage[pdftex]{graphicx}
\usepackage{amsmath,amsfonts,amsthm}
\usepackage{algorithmic}
\usepackage{array}
\usepackage[caption=false,font=footnotesize]{subfig}
\usepackage{fixltx2e}
\usepackage{stfloats}
\usepackage{authblk}
\usepackage{url}
\usepackage{color}
\usepackage{slashbox}
\usepackage[top=1.in, bottom=1in, left=1.1in, right=1.1in]{geometry}

\newcommand{\Tr}{\operatorname{Tr}}
\newcommand{\bR}{\mathbb{R}}
\newcommand{\bC}{\mathbb{C}}
\newcommand{\bP}{\mathbb{P}}
\newcommand{\bE}{\mathbb{E}}

\newcommand{\cI}{\mathcal{I}}
\newcommand{\cY}{\mathcal{Y}}
\newcommand{\cM}{\mathcal{M}}
\newcommand{\cP}{\mathcal{P}}
\newcommand{\cB}{\mathcal{B}}
\newcommand{\cO}{\mathcal{O}}
\newcommand{\hpi}{\hat{\pi}}
\newcommand{\ha}{\hat{a}}
\newcommand{\hm}{\hat{m}}
\newcommand{\hp}{\hat{p}}
\newcommand{\hb}{\hat{b}}
\newcommand{\dist}{\operatorname{dist}}

\newcommand{\rev}[1]{{\color{black}{#1}}}
\newcommand{\rebuttal}[1]{{\color{black}{#1}}}
\newtheorem{theorem}{Theorem}[section]
\numberwithin{equation}{section}

\providecommand{\keywords}[1]{\textbf{\textit{Keywords: }} #1}

\begin{document}

\title{Heterogeneous multireference alignment for images\\with application to 2-D classification\\in single particle reconstruction}

\author[1]{Chao Ma}
\author[1]{Tamir Bendory}
\author[2]{Nicolas Boumal}
\author[3]{Fred Sigworth}
\author[1,2]{Amit Singer}
\affil[1]{The Program in Applied and Computational  Mathematics, Princeton University, Princeton, NJ,  USA}
\affil[2]{Department of Mathematics, Princeton University, Princeton, NJ,  USA}
\affil[3] {Department of Cellular and Molecular Physiology, Yale University School of Medicine, New Haven, CT, USA}

\maketitle

\begin{abstract}
Motivated by the task of $2$-D classification in single particle reconstruction  by cryo-electron microscopy (cryo-EM), we consider the problem of heterogeneous multireference alignment of images. In this problem, the goal is to estimate a (typically small) set of target images from a (typically large) collection of observations. Each observation is a rotated, noisy version of one of the target images. For each individual observation, neither the rotation nor which target image has been rotated are known. As the noise level in cryo-EM data is high, clustering the observations and estimating individual rotations is challenging. We propose a framework to estimate the target images directly from the observations, completely bypassing the need to cluster or register the images. The framework consists of two steps. First, we estimate rotation-invariant features of the images, such as the bispectrum. These features can be estimated to any desired accuracy, at any noise level, provided sufficiently many  observations are collected. Then, we estimate the images from the invariant features. Numerical experiments on synthetic cryo-EM datasets demonstrate the effectiveness of the method. Ultimately, we outline future developments required to apply this method to experimental data.
\end{abstract}



\section{Introduction}

Single particle reconstruction using cryo-EM is a high-resolution imaging technique used in structural biology to image 3-D structures of macromolecules~\cite{frank2006three,kuhlbrandt2014resolution}. 
In a cryo-EM experiment, multiple samples of a particle are frozen in a thin layer of vitreous ice. 
Within the ice, the samples are randomly oriented and positioned. The electron microscope produces a tomographic image of the ice and the embedded samples, called a micrograph. The goal is then to estimate the 3-D structure of the particle from the micrograph. 
Importantly, the signal to noise ratio (SNR) of the  micrograph is usually low because of the limited electron dose that can be applied without causing excessive radiation damage. 

The first stage of existing cryo-EM algorithmic pipelines  is called \emph{particle picking}. In this stage, one aims to 
detect the projections of the samples within the micrograph and extract them. 
We  refer to these extracted images as \emph{projection images}. 
Throughout this paper, we assume perfect particle picking, that is, we obtain a large number of projection images, each containing a centered particle projection corresponding to an unknown viewing direction.

An important intermediate stage in the 3-D reconstruction procedure of cryo-EM is called \emph{2-D classification}. 
The goal of this stage is to produce $K$ 2-D images---called \emph{class averages}---with higher SNR. 
Each one of the $K$ images should represent a subset of the projections taken from a similar viewing direction.
The 2-D class averages can be used as templates for particle picking, to construct {\em ab initio} \mbox{3-D} structures~\cite{singer2011three, Christian2018Multicolor, van1987angular, goncharov1987three}, to provide a quick
assessment of the particles, \rev{to remove picked particles which are associated with non-informative
classes, and for symmetry detection~\cite{dube1993portal,singer2018mathematics,bendory2019single}.}

Different solutions were proposed for the 2-D classification problem. 
One approach is the reference-free alignment (RFA) technique~\cite{penczek1992three}. 
RFA tries to align all projection images globally by estimating all individual rotation parameters. However, when the images arise from many different viewing directions, RFA tends to produce large errors as no assignment of in-plane rotational angles can align all images simultaneously.
\rev{Methods based on expectation-maximization (EM)---an iterative algorithm that aims to find the marginalized maximum likelihood---are also popular}.  
In the context of cryo-EM, the method is usually referred to as Maximum Likelihood 2-D classification (ML2D). The method was first proposed in~\cite{sigworth1998maximum}, and is implemented in the popular software package RELION~\cite{scheres2012relion,scheres2005maximum}. 
Nevertheless, the EM framework lacks theoretical analysis and may be computationally expensive. \rev{In addition,  EM 
suffers from an intrinsic resolution-computational load trade-off, since the sampling of the in-plane rotation angles is Nyquist sampling.}
In Section~\ref{sec:num} we present some numerical results of EM and discuss more of its properties.

\rev{A different 2-D classification technique is based on multireference alignment (MRA)~\cite{van1996new,sorzano2010clustering,yang2012iterative}.} 
In MRA, the images are clustered into $K$ classes and the images within each class are averaged to suppress the noise. The averaged images are the class averages. 
As projection images can be similar up to rotation, the clustering is based \rev{ on either rotationally aligning the images within each class, or on features of the images which are invariant under rotations, such as  autocorrelation~\cite{schatz1990invariant} or bispectrum~\cite{zhao2014rotationally}.} 
MRA and invariant features play a key role in this paper and are discussed in detail later. Notably, our proposed method avoids the clustering stage, which may be inaccurate at low SNR.
Instead, we aim to estimate the $K$ class averages directly  from the projection images, with no intermediate clustering stage. 

In this paper, we propose to model the 2-D classification problem  as an instance of the \emph{heterogeneous multireference alignment} (hMRA) problem, for the case of 2-D images~\cite{perry2019sample,boumal2018heterogeneous}. 
In the hMRA problem, we observe $N$ projection images. Each observed image 
is an in-plane rotated, noisy version of one of the $K$ underlying images---the class averages, corresponding to $K$ viewing directions. 
For each observation, the specific underlying image and the in-plane rotation are unknown.
Since in cryo-EM all in-plane rotations are equally likely to appear, we assume uniform distribution of  rotations. 
The goal is to estimate the $K$ class averages, as well as the distribution among the class averages. 
Crucially, for each observed particle, which class average it came from (its label) and which in-plane rotation was applied to it are treated as \emph{nuisance variables}; that is: they are unknown, but we do not seek to estimate them.
A detailed mathematical model of the hMRA problem is provided in Section~\ref{sec:stat}. 

In the high SNR regime, the nuisance variables could be estimated accurately, at least in principle. 
Given an accurate estimate of these variables, the problem becomes trivial: one can cluster the observations into the $K$ class averages,  undo the rotations, and average within each class to suppress the noise. However, in the low SNR regime, estimating the labels and rotations becomes challenging, and indeed impossible as the SNR drops to zero; see for instance~\cite{aguerrebere2016fundamental} for analysis in a related model. Notwithstanding, it was shown in a series of papers that in many MRA setups the underlying signal (or signals in our case) can be estimated at any noise level, provided sufficiently many observations are recorded~\cite{bandeira2014multireference,bendory2017bispectrum,bandeira2017optimal,
perry2019sample,abbe2017sample,abbe2017multireference,bandeira2017estimation}. Remarkably, it was shown that \rev{in the low SNR regime}, the method of moments  achieves the optimal sample complexity under rather moderate conditions~\cite{abbe2018estimation,bandeira2017optimal}.

Consequently, targeting the low SNR regime, we propose to apply the method of moments of MRA to the 2-D classification problem in cryo-EM.
Our work builds upon the notion of~\emph{bispectrum}, first proposed by Tukey~\cite{tukey1953spectral}, and currently used in signal processing~\cite{giannakis1989signal,sadler1992shift}. 
The bispectrum is invariant under  rotations; that is, the bispectrum of an image remains unchanged
after an arbitrary in-plane rotation. This property enables us to bypass  estimation of individual rotations associated with each one of the observations. Under the assumption of uniform distribution of rotations, the bispectrum is equivalent to the third-order moment of the image. 
Inspired by the seminal work of Kam~\cite{kam1980reconstruction},   
previous works~\cite{bendory2017bispectrum,boumal2018heterogeneous} studied the MRA and hMRA problems for 1-D signals using the bispectrum as a simplified model for the 3-D reconstruction problem in cryo-EM. In this paper, we study the more involved hMRA model for 2-D images as a model for 2-D classification.

In a nutshell, our proposed approach for 2-D classification consists of the following stages. 
First, we expand each image in a  \emph{steerable basis}. In this paper, we use the Fourier-Bessel basis,
but alternative bases, such as the prolate spheroidal wave functions, can be  used alternatively~\cite{landa2017steerable}.  
As explained in Section~\ref{sec:stat}, in such a basis all in-plane rotations of an image admit the same expansion coefficients, up to  complex phase modulations. This property is called \emph{steerability}.
As a result, specific monomials in these coefficients are invariant under in-plane rotations. We refer to these monomials as \emph{invariant features}. Specifically, we make use of monomials of the first-, second- and third-order called  the mean, power spectrum and bispectrum, respectively. 

In practice, instead of working directly on the expansion coefficients, we employ a dimensionality-reduction and denoising technique called \emph{steerable principal component analysis} (sPCA). This technique is similar to the standard PCA, while boosting the SNR by accounting for all in-plane rotations of the data in an efficient way~\cite{zhao2016fast}. In addition, the sPCA coefficients preserve the steerability property. Therefore, the invariant monomials can be computed in the sPCA space. 

After computing the invariant features of each image, we average over all images. These averages are consistent estimators (up to bias terms that can be removed easily) of the mixed invariant features of the $K$ class averages. Ultimately, a nonconvex least squares (LS) optimization problem is designed to recover the sPCA coefficients of the individual class averages from these mixed invariant features. 
All the ingredients of this algorithmic pipeline are provided in Sections~\ref{sec:stat} and~\ref{sec:alg}.
 Figure~\ref{fig:flowchart} illustrates the flowchart of the procedure. Numerical results and comparison with EM are provided in Section~\ref{sec:num}. Section~\ref{sec:con} concludes this work, discusses its limitations and potential future extensions.

\section{Statistical model and invariant features}\label{sec:stat}
In this section, we first describe in detail the hMRA model. Then, we introduce our framework based on computing features that are invariant under rotations in a steerable basis.

\subsection{Statistical model}\label{ssec:model}
Let $\{I_1, I_2, ..., I_K\}$ be a set of $K$ images of size $L\times L$, with $L$ odd\footnote{We consider an odd $L$  for convenience of implementation and presentation. Our algorithm also works when $L$ is even.} and pixel values in $[0,1]$: these are the class averages, our target parameters.
The pixels in an image are indexed by a pair of integers $(x,y)$  with $-(L-1)/2\leq x,y\leq(L-1)/2$. 
The support of  the images is assumed to lie in the disk $x^2+y^2\leq(L-1)^2/4$; as a result, any of their rotations have the same property. Let $R_\theta$ be a rotation operator which rotates an image  counter-clockwise by angle $\theta$, and let $\xi$ be a random variable following a uniform distribution on $[0,2\pi)$.
In addition, let $\pi$ be a random variable on the set $\{1,2,...,K\}$ with distribution $(\pi_1, \ldots, \pi_K)$:
\begin{equation*}
\pi_k:=\bP(\pi=k)> 0,\ \  k\in\{1,2,...,K\}.
\end{equation*}
Then, our observations are i.i.d.\ random samples from the model 
\begin{equation}\label{eq:model}
Y=R_{\xi}I_\pi+\varepsilon,
\end{equation}
where $\varepsilon=(\varepsilon_{ij})\in\bR^{L\times L}$ is a noise matrix of i.i.d.\ Gaussian variables with zero mean and variance $\sigma^2$; the random variables $\xi, \pi, \varepsilon$ are independent. Indeed, it was observed that the background noise in cryo-EM experiments can be treated as Gaussian~\cite{park2011stochastic}.

Suppose we collect $N$ independent observations from the generative model~\eqref{eq:model},
\begin{equation*}
\cY=\{Y_1, Y_2,..., Y_N\},
\end{equation*}
so that 
\begin{equation*}
Y_i = R_{\xi_i}I_{\pi_i}+\varepsilon_i, \quad i=1,\ldots,N.
\end{equation*}
From the observed data $\cY$, we seek to estimate the target images $\{I_1, I_2,..., I_K\}$ (class averages) and, possibly, the distribution $\pi$, without estimating the in-plane rotations of individual observations ${\xi_i}$ or the labels $\pi_i$.
The model~\eqref{eq:model} suffers from unavoidable ambiguities of rotations (of each class average) and permutation (across the $K$ images). Therefore, naturally, a solution is defined up to these symmetries. In Section~\ref{sec:num} we define a suitable error metric. 

\subsection{Steerable basis}\label{ssec:inv}
As mentioned above, we aim to bypass estimating the nuisance variables by computing features that are invariant under rotation. To this end, we first expand the images with respect to a steerable basis. In polar coordinates, steerable basis functions take the form
\begin{equation}\label{eq:sepa}
u^{k,q}(r,\theta)=f^{k,q}(r)e^{\iota k\theta},
\end{equation}
where $\iota:=\sqrt{-1}$. Notice the separation of variables: If we expand an image in $u^{k,q}$, 
\begin{equation}
I(r,\theta)=\sum\limits_{k,q}a_{k,q}u^{k,q}(r,\theta),
\end{equation}
then the expansion of the rotated image follows from:
\begin{equation}\label{eq:rot}
\begin{split}
(R_\alpha I)(r, \theta) = I(r,\theta-\alpha)&=\sum\limits_{k,q}a_{k,q}u^{k,q}(r,\theta-\alpha)\\ &=\sum\limits_{k,q}a_{k,q}e^{-\iota k \alpha}u^{k,q}(r,\theta).
\end{split}
\end{equation}
Since our images are real, the coefficients satisfy a conjugacy symmetry: $a_{k,q}=\overline{a_{-k,q}}$. Therefore, coefficients with $k\geq0$ suffice to represent the images. 

Examples of steerable bases include the Fourier-Bessel basis and prolate spheroidal wavefunctions. See \cite{zhao2016fast}\cite{landa2017approximation}\cite{landa2017steerable}\cite{lederman2017numerical} for efficient expansion algorithms. In this paper, we work with the Fourier-Bessel basis on a disk with radius $c$ defined as:
\begin{equation}
u^{k,q}(r,\theta)=\left\{
\begin{array}{cc}
N_{k,q}J_k(R_k\frac{r}{c})e^{\iota k\theta}, & r\leq c,\\
0, & r>c,\\
\end{array}\right.
\end{equation}
where $J_k$ is the Bessel function of the first kind, $R_{k,q}$ is the $q^{th}$ root of $J_k$ and  
$N_{k,q}=(c\sqrt{
\pi}|J_{k+1}(R_{k,q})|)^{-1}$  is a normalization factor. We take $c$ to be $(L-1)/2$, in accordance with the assumed support of the images. 

To reduce the dimensionality of the representation and  denoise the image, we perform sPCA after expanding the images in a steerable basis~\cite{zhao2016fast}.
The sPCA results in a new, data driven basis to represent the images. Importantly,
this new basis preserves the steerability property and consequently the rotation property~\eqref{eq:rot} holds true. 
Section~\ref{ssec:sPCA}  introduces the sPCA technique in more details. 
With some abuse of notation, in what follows the coefficients in a sPCA basis are also denoted by $\{a_{k,q}\}$. 

\subsection{Invariant features}

The steerability property~\eqref{eq:rot}  enables us to determine features of images which are invariant under rotation.
Specifically, features that are invariant to an action of $SO(2)$: the special orthogonal group in 2-D.
 We assume that the images are ``band-limited'' in the sense that their expansion in a steerable basis is finite.

From~\eqref{eq:rot}, it is clear that coefficients $a_{k,q}$ corresponding to $k=0$ are not affected by rotation. Hence, the first-order invariants are just the  mean values, or the ``DC components'':
\begin{equation}\label{eq:mom1}
m_q=a_{0,q},
\end{equation}
for all $q$. The second-order invariants, which form the power spectrum, are given by
\begin{equation}\label{eq:mom2}
p_{k,q_1,q_2}=a_{k,q_1}\overline{a_{k,q_2}},
\end{equation}
for all $k,q_1,q_2$. 
 The power spectrum coefficients are invariant to rotation since for all $\alpha$:
\begin{eqnarray*}
	&&\left(a_{k,q_1}e^{-\iota k\alpha}\right)\left(\overline{a_{k,q_2}e^{-\iota k \alpha}}\right)= a_{k,q_1}\overline{a_{k,q_2}}.
\end{eqnarray*}
Unfortunately, the power spectrum does not determine the image uniquely: a multiplication of the expansion coefficients by 
\rev{$e^{\iota h[k]}$ for an arbitrary function $h$} does not change the power spectrum, yet it does change the image.

The third-order invariant, the bispectrum, is defined by 
\begin{equation}\label{eq:mom3}
b_{k_1,k_2,q_1,q_2,q_3}=a_{k_1,q_1}a_{k_2,q_2}\overline{a_{k_1+k_2,q_3}},
\end{equation}
for all $k_1,k_2,q_1,q_2,q_3$.
Using equation~\eqref{eq:rot}, one can verify that indeed:
\begin{eqnarray*}
&&\left(a_{k_1,q_1}e^{-\iota k_1\alpha}\right)\left(a_{k_2,q_2}e^{-\iota k_2\alpha}\right)\left(\overline{a_{k_1+k_2,q_3}e^{-\iota (k_1+k_2)\alpha}}\right) \\
&=&\left(a_{k_1,q_1}a_{k_2,q_2}\overline{a_{k_1+k_2,q_3}}\right)\left(e^{-\iota k_1\alpha}e^{-\iota k_2\alpha}\overline{e^{-\iota (k_1+k_2)\alpha}}\right)\\
&=&a_{k_1,q_1}a_{k_2,q_2}\overline{a_{k_1+k_2,q_3}}.
\end{eqnarray*}
The combined power spectrum and bispectrum do determine an image uniquely, up to global rotation: 
\begin{theorem}
Consider two images with steerable basis coefficients $a_{k,q}$ and $a'_{k,q}$, respectively, in the range $-k_{max}\leq k\leq k_{max}$. Assume that for all $-k_{max}\leq k\leq k_{max}$ there exists at least one $q$ such that $a_{k,q}\neq0$. 
If for all indices
\begin{eqnarray}
&a_{k,q_1}\overline{a_{k,q_2}}=a'_{k,q_1}\overline{a'_{k,q_2}}, \label{eq:thm2} \\
&a_{k_1,q_1}a_{k_2,q_2}\overline{a_{k_1+k_2,q_3}}=a'_{k_1,q_1}a'_{k_2,q_2}\overline{a'_{k_1+k_2,q_3}}, \label{eq:thm3} 
\end{eqnarray}
then there exists $\theta\in[0,2\pi)$ such that
\begin{equation}
a'_{k,q}=a_{k,q}e^{-\iota k\theta}.
\end{equation}
for all $k, q$.
That is, the two images only differ by a rotation.
\end{theorem}
\begin{proof}
Set $q_1=q_2$ in~\eqref{eq:thm2}, we have $|a_{k,q}|=|a'_{k,q}|$ for any $k$ and $q$. Hence, $a'_{k,q}\neq0$ if and only if $a_{k,q}\neq0$, and there exists $\theta_{k,q}\in[0,2\pi)$ such that $a'_{k,q}=a_{k,q}e^{-\iota\theta_{k,q}}$. Then, still by~\eqref{eq:thm2}, we have
\[
a'_{k,q_1}\overline{a'_{k,q_2}}=a_{k,q_1}\overline{a_{k,q_2}}e^{-\iota(\theta_{k,q_1}-\theta_{k,q_2})}.
\]
This means that, for fixed $k$, $\theta_{k,q}$ take a same value (in $[0,2\pi)$) for all $q$ satisfying $a_{k,q}\neq0$.
Hence, for each $k$, there exists a single $\theta_k\in[0,2\pi)$ such that
\[
a'_{k,q}=a_{k,q}e^{-\iota\theta_k}.
\]
Next, by~\eqref{eq:thm3}, we have for all $k_1$, $k_2$, $q_1$, $q_2$, and $q_3$, 
\begin{align*}
a'_{k_1,q_1}a'_{k_2,q_2}\overline{a'_{k_1+k_2,q_3}}&=a_{k_1,q_1}a_{k_2,q_2}\overline{a_{k_1+k_2,q_3}}\\ &\times e^{-\iota\left(\theta_{k_1}+\theta_{k_2}-\theta_{k_1+k_2}\right)}.\end{align*}
By assumption, we can always choose $q_1,q_2$ and $q_3$ such that
	 $a_{k_1,q_1}$, $a_{k_2,q_2}$ and $a_{k_1+k_2,q_3}$ are all nonzero. Then, we have
\[
\theta_{k_1}+\theta_{k_2}=\theta_{k_1+k_2},
\]
for all $-k_{max}\leq k_1,k_2\leq k_{max}$.
This in turn implies 
\[
\theta_k=k\theta,
\]
for some constant $\theta$. Thus we conclude that
\[
a'_{k,q}=a_{k,q}e^{-\iota k\theta}. \qedhere
\]
\end{proof}

\section{$2$-D classification using invariant features}\label{sec:alg}

In this section, we introduce the whole pipeline of our algorithm in detail. We start by expanding the observed images in Fourier-Bessel basis. We then perform sPCA on the resulting coefficients. Next, we estimate the mixed invariants of the $K$ true images (class averages) from the  observations' sPCA coefficients. Ultimately, we estimate the sPCA coefficients of the true images, and thus the images themselves, from the mixed invariants  via a nonconvex LS optimization problem.
Figure~\ref{fig:flowchart} shows a schematic flow chart of our algorithm and Algorithm~\ref{alg:1} describes our algorithm step by step. Next, we elaborate on each of the steps.

\begin{figure}[!t]
\centering
\includegraphics[width=3in]{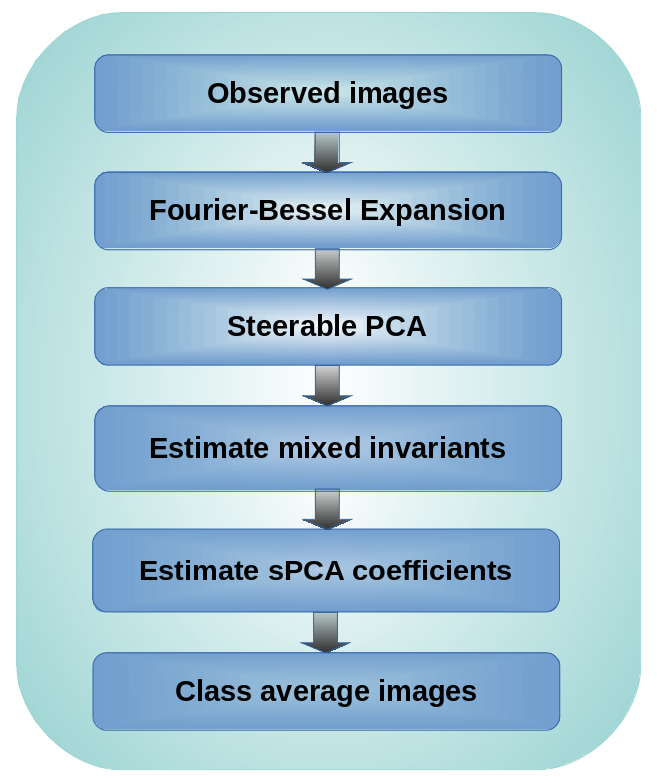}
\caption{Flow chart of Algorithm~\ref{alg:1} for $2$-D classification using rotationally invariant features.}
\label{fig:flowchart}
\end{figure}

\begin{algorithm2e}\caption{2-D classification by invariant features}\label{alg:1}
\textbf{Input: } Observations $Y_1, \ldots, Y_N$; noise variance $\sigma^2$.\\
Expand the observations in the Fourier-Bessel basis, and perform sPCA on the expansion coefficients using the method described in~\cite{zhao2016fast}. \\
Estimate the mixed invariants of the true images using the sPCA coefficients of the data by~\eqref{eq:est1}, \eqref{eq:est2} and~\eqref{eq:est3}. \\
Estimate the sPCA coefficients of the true images and the distribution $\pi$ by solving the optimization problem~\eqref{eq:obj}. \\
Recover the images $\hat{I}_1, \ldots, \hat{I}_K$ from the sPCA coefficients by~\eqref{eq:rec}.\\
\textbf{Output: } images $\hat{I}_1, \ldots, \hat{I}_K$ (up to permutation and rotations); distribution $\hat{\pi}$.
\end{algorithm2e}

\subsection{Fourier-Bessel sPCA}\label{ssec:sPCA}
We perform sPCA on the images after they were expanded in a Fourier-Bessel steerable basis, introduced in Section \ref{ssec:inv}. Like in a standard PCA, the first step is to subtract the mean observed image from each observation to center the data. The mean image is added back at the last step of the PCA. To ease exposition, we assume (only in this subsection) that the images have zero mean.

For a regular PCA, one would construct the data matrix $X$ such that each column holds the expansion coefficients of one image. Then, PCA would extract the dominant eigenvectors of \rev{$XX^*$, where $X^*$ is the conjugate transpose of $X$. Note that since $XX^*$ is Hermitian, it has real eigenvalues.
} 
Since in our model (and in cryo-EM) any observed images could have been observed after arbitrary in-plane rotation with the same probability, we wish to include all such rotated versions of all images in the PCA procedure. This can be done efficiently owing to steerability, as described in \cite{zhao2013fourier}. The dominant eigenimages obtained through sPCA form an orthonormal basis for a lower dimensional subspace, where we now project all observations. Crucially, this eigenbasis is also steerable (since each eigenimage is a linear combination of steerable basis functions.) We still get the two usual benefits of PCA---dimensionality reduction and denoising---with the added benefit that we exploited all of the available information. \rev{The sPCA has been proven to be an effective denoising tool for cryo-EM  reconstruction~\cite{zhao2014rotationally,bhamre2016denoising,levin20183d,landa2017steerable}.}
The  procedure is actually faster than standard PCA because the data covariance $XX^*$ is block-diagonal upon factoring in all rotated images. In the next subsection, we use the expansion coefficients in the sPCA basis to compute invariant features.

\subsection{Estimating the invariant features of the class average  images} 
After performing sPCA, we get a steerable eigenbasis (a collection of eigenimages) and the expansion coefficients of the observed images in that basis (after projection to the corresponding subspace). Then, we can compute the  invariant features using these coefficients. The  invariants can be computed according to equations~\eqref{eq:mom1}, \eqref{eq:mom2} and~\eqref{eq:mom3}. Next, we estimate the mixed invariants of the underlying class average images using the invariants of the noisy data, which we now explain.

Let $\{a^i_{k,q}\}$ be the sPCA coefficients of the $i$th target image $I_i$. 
As per our model~\eqref{eq:model}, in the absence of noise, the coefficients of an observation $Y$ are given by $a^\pi_{k,q}e^{-\iota k\xi}$. Let $b_{k_1,k_2,q_1,q_2,q_3}^{\pi,\xi}$ be the bispectrum computed from the latter. By construction, this is independent of $\xi$: this is simply the bispectrum of the target image $I_\pi$. Marginalizing over the remaining nuisance variable $\pi$, we find
\begin{eqnarray}
\bE_{\pi} b_{k_1,k_2,q_1,q_2,q_3}^{\pi,\xi} &=& \bE_{\pi} a^\pi_{k_1,q_1}a^\pi_{k_2,q_2}\overline{a^\pi_{k_1+k_2,q_3}} \label{eq:avg1},
\end{eqnarray}
where $\bE_{\pi}$ represents expectation taken against $\pi$: a sum over $\pi = 1, \ldots, K$ weighted by $(\pi_1, \ldots, \pi_K)$.

This relation implies that by  averaging over all the bispectra of the observations  we can estimate the mixed bispectra of the $K$ class averages. Estimation of mixed mean and power spectra can be similarly obtained.
Crucially, we approximate the mixed invariants of the true images without estimating the rotations $\xi$ or the labels $\pi$ of individual  observations.

The same method can be applied in the presence of noise.
Now, the coefficients of the observations are given by 
\begin{equation}
a^\pi_{k,q}e^{-\iota k\xi}+\varepsilon_{k,q}^ce^{-\iota k\xi},
\end{equation}
where $\varepsilon^c$ denotes the complex Gaussian noise in the coefficients, satisfying $\bE \varepsilon^c=0$ and $\bE \varepsilon^c(\varepsilon^c)^*=\sigma^2I$, \rev{where $I$ is the identity matrix.} 
The noise terms induce bias in the power spectrum and bispectrum estimation. Particularly,  the  noisy power spectrum of $Y$  
satisfies
\begin{eqnarray}\label{eq:powerspectrum_noise}
\bE_{\pi,\varepsilon} p_{k,q_1,q_2}^{\pi,\xi,\varepsilon}&=&\bE_{\pi} a_{k,q_1}^{\pi}\overline{a_{k,q_2}^\pi}+\sigma^2\delta_{q_1,q_2},
\end{eqnarray}
where $\delta_{q_1,q_2}$ is the Kronecker delta function. 
Hence, we get a bias term which depends solely on $\sigma$,   which is usually estimated in the cryo-EM algorithmic pipeline. 
Similarly, expectation over the noisy bispectrum results in  
\begin{equation} \label{eq:bispectrum_noise}
\begin{split}
\bE_{\pi,\varepsilon} b_{k_1,k_2,q_1,q_2,q_3}^{\pi,\xi,\varepsilon} &=\bE_\pi a^\pi_{k_1,q_1}a^\pi_{k_2,q_2}\overline{a^\pi_{k_1+k_2,q_3}} \\ &+\sigma^2\bE_\pi A^\pi,
\end{split}
\end{equation}
where 
\begin{equation*}
A^{\pi}:=\delta_{q_2,q_3}\delta_{k_1,0}a^\pi_{0,q_1}
+\delta_{q_1,q_3}\delta_{k_2,0}a^\pi_{0,q_2}  
+\delta_{q_1,q_2}\delta_{k_1+k_2,0}a^\pi_{0,q_3}.
\end{equation*}
Here, the bias term depends on both $\sigma^2$ and the coefficients $a_{0,q}^\pi$. Noise does not introduce bias in estimates of the mean.

Equipped with~\eqref{eq:powerspectrum_noise} and~\eqref{eq:bispectrum_noise}, estimating the mixed invariants can be executed by averaging  over the invariants of the observations and removing  the bias terms. 
Specifically, let $m^{Y_i}_q$, $p^{Y_i}_{k,q_1,q_2}$, $b^{Y_i}_{k_1,k_2,q_1,q_2,q_3}$ be,  respectively, the mean, power spectrum and bispectrum of the noisy observation $Y_i$. 
Then, our estimators of the mixed invariants are easily computed as:
\begin{align}
\hat{m}_{q}=&\frac{1}{N}\sum\limits_{i=1}^N m^{Y_i}_q, \label{eq:est1}\\
\hat{p}_{k,q_1,q_2}=&\frac{1}{N}\sum\limits_{i=1}^N p^{Y_i}_{k,q_1,q_2} - \sigma^2\delta_{q_1,q_2},\label{eq:est2}\\
\hat{b}_{k_1,k_2,q_1,q_2,q_3}=&\frac{1}{N}\sum\limits_{i=1}^N b^{Y_i}_{k_1,k_2,q_1,q_2,q_3}-\sigma^2 E_\pi A^\pi.\label{eq:est3}
\end{align}
In~\eqref{eq:est3}, the bias term $E_\pi A^\pi$ can be estimated by $\hat{m}_q$.

\subsection{Estimating the coefficients of the class averages}
In the last section, we showed how the  mixed invariants of the class averages can be estimated from the data. Now, we turn our attention to estimating the sPCA coefficients of the class average images from their mixed invariants by a LS optimization problem.

 Our optimization problem consists of two types of variables. The first type 
represents the sPCA coefficients of the target images:
\begin{equation}
\hat{a}^i=\{\hat{a}^i_{k,q}\},\ \ i=1,2,...,K.
\end{equation}
The second type represents the distribution from which the observations are sampled:
\begin{equation}
\hat{\pi}=(\hat{\pi}_1,\ldots\hat{\pi}_K).
\end{equation} 
 In practice,  as long as we have sufficiently many observations so that the empirical estimates of the invariant features are accurate, the following identities hold approximately:
\begin{align}
\sum\limits_{i=1}^K \pi_i a^i_{0,q} &\approx \hat{m}_{0,q}, \label{eq:approx1}\\
\sum\limits_{i=1}^K \pi_i a^i_{k,q_1}\overline{a^i_{k,q_2}} &\approx \hat{p}_{k,q_1,q_2}, \label{eq:approx2}\\
\sum\limits_{i=1}^K \pi_i a^i_{k_1,q_2}a^i_{k_2,q_2}\overline{a^i_{k_1+k_2,q_3}} &\approx \hat{b}_{k_1,k_2,q_1,q_2,q_3}.\label{eq:approx3}
\end{align}
 In an ideal case, we want $\hat{a}^i$ to be coefficients of $I_i$ (or an in-plane rotated version thereof), and $\hat{\pi}_i=\pi_i$.
Hence, we design an LS problem to minimize the difference between left- and right-hand sides of equations~\eqref{eq:approx1}, \eqref{eq:approx2} and~\eqref{eq:approx3}.
Let $M_q(\hat{a},\hat{\pi})$, $P_{k,q_1,q_2}(\hat{a},\hat{\pi})$, $B_{k_1,k_2,q_1,q_2,q_3}(\hat{a},\hat{\pi})$ capture the left-hand sides of the three equations above, respectively, with the estimators $(\hat{a},\hat{\pi})$ rather the unknown,  underlying parameters  $({a},{\pi})$.
Our objective function reads 
\begin{equation} \label{eq:obj}
\begin{split}
&F(\hat{a},\hat{\pi}) = \sum\limits_{q}\left|M_q(\hat{a},\hat{\pi})-\hat{m}_{0,q}\right|^2 \\
&+ \frac{1}{1+\sigma^2} \sum\limits_{k,q_1,q_2}\left|P_{k,q_1,q_2}(\hat{a},\hat{\pi})-\hat{p}_{k,q_1,q_2}\right|^2 \\
&+ \frac{1}{1+\sigma^2+\sigma^4} \sum_{\substack{k_1,k_2,\\q_1,q_2,q_3}}\left|B_{k_1,k_2,q_1,q_2,q_3}(\hat{a},\hat{\pi})-\hat{b}_{k_1,k_2,q_1,q_2,q_3}\right|^2.
\end{split}
\end{equation}
Here we use $1$, $1+\sigma^2$ and $1+\sigma^2+\sigma^4$ 
as  rough estimates for  the variances of corresponding terms; see also~\cite{boumal2018heterogeneous}.

Some constraints need to be imposed. First, as variables $\hat{\pi}$ are used to represent a distribution, they should lie on the simplex, that is, $\hat{\pi}_i \geq 0$ and 
$\hat\pi_1 + \cdots + \hat\pi_K = 1$.
Second, as our images are real, and Fourier-Bessel sPCA basis functions with $k=0$ are real, we have $\Im\left(a^i_{0,q}\right) = 0$ (where $\Im$ extracts imaginary part). Consequently, we can force $\hat{a}^i_{0,q}$ to be real. Similarly, coefficients with the same $q$ but opposite $k$ are conjugate,
\begin{equation}
\hat{a}^i_{k,q}=\overline{\hat{a}^i_{-k,q}}.
\end{equation}
In the optimization problem, we can just consider those coefficients with nonnegative $k$. 

To conclude, we aim to solve the following constrained LS optimization problem:
\begin{eqnarray} \label{eq:LS}
&\min\limits_{\hat{a}^i_{k,q}\in\bC,\ \hat{\pi}_i\in\bR} & F(\hat{a},\hat{\pi}), \label{eq:opt}\\
&\mbox{subject to}  & \sum\limits_{i=1}^K \hat{\pi}_i=1,\ \ \hat{\pi}_i\geq0, \nonumber\\
&      & \Im\left(\hat{a}^i_{0,q}\right)=0. \nonumber
\end{eqnarray}
While the LS is nonconvex, we find that we can solve it satisfactorily in practice---see Section~\ref{sec:num}.
This is in line with recent related work~\cite{bendory2017bispectrum,abbe2017multireference,boumal2018heterogeneous,bendory2019multi}.

We attempt to solve the optimization problem~\eqref{eq:opt} by a trust-regions algorithm or conjugate gradient method using Manopt \cite{boumal2014manopt}: a toolbox for optimization on manifolds\footnote{\url{www.manopt.org}}. In our problem, the variable $\ha$ lies in Euclidean space, while $\hpi$ lies on the simplex, whose relative interior is endowed with a Riemannian geometry in the toolbox~\cite{sun2015multinomial}. When using Manopt, we only provide the gradient on Euclidean space. The gradient on the manifold is computed from the gradient on Euclidean space together with the representation of the manifold. The Hessian is approximated automatically by finite differences of the gradient.

\subsection{Recovery of the images}
After solving the optimization problem, we obtain a collection of coefficients $\hat{a}^i_{k,q}$, which are believed to approximate the sPCA coefficients of the target images, $a^i_{k,q}$.
To recover the images themselves, up to rotation,  we simply compute a linear combination of the basis images given by the sPCA, using coefficients $\hat{a}^i$. 
We then add the mean image that was  subtracted from all observations during the sPCA preprocessing, denoted by $I_m^i$---see Section~\ref{ssec:sPCA}.
Specifically, letting $\Phi$ be the sPCA basis, we recover the images by
\begin{equation}\label{eq:rec}
\hat{I}_i = \Phi \hat{a}^i+I_m^i. 
\end{equation}

\subsection{Computational complexity}
In this section we discuss the computational complexity of each step of Algorithm~\ref{alg:1}. By~\cite{zhao2016fast}, the computational cost of the sPCA step is $\cO(NL^3+L^4)$, where $N$ is the number of observations and $L$ is the side length of the images. For each image, assume the sPCA provides $M$ components and the maximum angular frequency is $k_{max}$. 
Then, we obtain  $\cO(\frac{M^3}{k_{max}})$ invariants in total~\cite{zhao2014rotationally}. Hence, $\cO(\frac{NM^3}{k_{max}})$ computations are required to compute the invariants of all observations and estimate the mixed invariants for groundtruth images. 
Next, in the optimization part, computing the gradient requires going through all terms in the objective function, and each term contributes $\cO(K)$ elements of the gradient. Hence, if $T$ iterations are performed, the computational complexity of the optimization step is $\cO(\frac{TKM^3}{k_{max}})$. 
 Finally, building the recovered images just involves $K$ linear combinations of the principal components with  computational cost $\cO(KML^2)$.

\section{Numerical experiments}\label{sec:num}

In this section, we show results of numerical experiments using our algorithm. First, we use random projections of the E.~coli $70$S ribosome volume~\cite{shaikh2008spider} as the groundtruth images to explore the performance of our algorithm under different noise levels and distributions. The volume is available in the software ASPIRE package\footnote{\url{www.spr.math.princeton.edu}}. The size of each image is $129^2$ \rev{(i.e., $L=129$)}. \rebuttal{Figure~\ref{fig:eg_E70s} shows some examples of the class averages and noisy input images.} Later, we apply our algorithm on two other molecules with projections of larger size. Code for our algorithm and all experiments is available at https://github.com/chaom1026/2DhMRA. The experiments presented below are conducted by MATLAB on a machine with $4\times$ E7-8880 v3 CPUs, and 750GB of RAM. 

\begin{figure}
\centering
\includegraphics[width=3.0in]{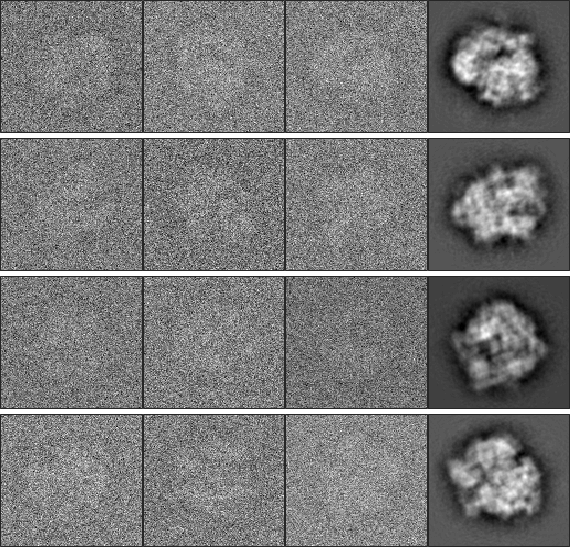}
\caption{\rebuttal{Examples of class averages (the \rev{right} column) and rotated noisy images (the first three columns) for the E.~coli $70$S ribosome volume. SNR=$1/50$.}}
\label{fig:eg_E70s}
\end{figure}

Following~\cite{boumal2018heterogeneous}, we define error metrics suitable for the inherent symmetries of our problem. For two images $I$ and $\hat{I}$, we define a rotationally invariant distance as 
\begin{equation}
\dist(I,\hat{I})=\min\limits_{\theta\in[0,2\pi]}\|R_\theta I-\hat{I}\|_{\text{F}}.
\end{equation}
This distance measures the Frobenius  norm between all rotational alignments of $I$ and $\hat{I}$. 
Let $\cI=\{I_1, ..., I_K\}$ be a set of underlying images, and let $\hat{\cI}=\{\hat{I_1}, ..., \hat{I}_K\}$ be the estimates. 
To be  invariant under both rotations and permutations, we use the following definition:
\begin{equation}\label{eq:disperm}
\dist(\cI,\hat{\cI})^2=\min_{p\in S_K} \sum\limits_{i=1}^K \dist(I_i,\hat{I}_{p(i)})^2,
\end{equation}
where $S_K$ is the set of all permutations of $\{1,2,...,K\}$. 
The relative error between $\cI$ and $\hat{\cI}$  is defined by
\begin{equation}
\dist_r(\cI,\hat{\cI})=\frac{\dist(\cI,\hat{\cI})}{\sqrt{\sum_{i=1}^K \|I_i\|_{\mathrm{F}}^2}}.
\end{equation}
\rev{Note that when computing errors we only consider the disk area with diameter $L$ and ignore the corners.} 
We measure the error between $\hat{\pi}$ and the true distribution $\pi$ by the total variation (TV) distance which takes values in $[0,1]$:
\begin{equation}
\dist_{\mathrm{TV}}(\hat{\pi},\pi)=\frac{1}{2}\sum_{i=1}^K |\hat{\pi}_i-\pi_i|.
\end{equation}
Here we assume that a permutation given by equation \eqref{eq:disperm} has been applied to $\hat{\pi}$.
In what follows, we define the SNR as
\begin{equation}
\mbox{SNR}=\frac{\bE(\mathrm{Signal}^2)}{\bE(\mathrm{Noise}^2)}:=\frac{\sum_{i=1}^K \|I_i\|_{\mathrm{F}}^2}{KL^2\sigma^2}.
\end{equation}

During sPCA, we use the method introduced in~\cite{zhao2016fast} to choose the eigenimages automatically, based on \rev{properties of the Marchenko-Pastur distribution}. Specifically, for each frequency $k$, we take those eigenimages with eigenvalues satisfying 
\begin{equation} \label{eq:mp}
\lambda^{(k)}>1.005\sigma^2(1+\sqrt{\gamma_k}),
\end{equation}
where $\sigma^2$ is the variance of the noise, $\gamma_0=\frac{p_0}{N}$, and $\gamma_k=\frac{p_k}{2N}$ for $k\neq0$. Here,  $p_k$ is the number of eigenimages for frequency $k$ and $N$ is the total number of observations. The \rev{factor} $1.005$ in~\eqref{eq:mp} is chosen heuristically to control the number of sPCA coefficients.

In our first experiment we consider a uniform distribution $\pi$, and assume that the algorithm knows that $\pi$ is uniform. Hence, in the optimization problem, variables $\hat{\pi}_i$ are fixed to be $1/K$. We choose $K=10$, $\mbox{SNR}=1/50$, and take $N=10^4$ observations in total. Examples of noisy observations are shown in the third column of Figure \ref{fig:uniform}. With this much noise, it would be challenging to rotationally align and cluster the observations. 
Nevertheless, the proposed algorithm gets estimates of the images without (even implicitly) doing either alignment or clustering. During sPCA, $83$ coefficients are automatically chosen to represent each image, according to~\eqref{eq:mp}. Hence, in total we have $830$ variables. Figure~\ref{fig:uniform} shows some examples of original images (before and after sPCA), noisy observations and the recovered images by our algorithm. We can see that our algorithm produces accurate recovery of the groundtruth images from noisy samples.
We split the measured error into two terms: the error caused by sPCA (the error between groundtruth images before and after sPCA) and estimation error in the sPCA space caused by the optimization problem (the error between recovered images and groundtruth images after sPCA). 
We refer to these errors as sPCA error and estimation error, respectively. 
 For this experiment, the relative sPCA error compared to the groundtruth images is about $19.6\%$, while the relative estimation error compared to groundtruth images after sPCA is about $ 5.2\%$. Table~\ref{tab:time} shows the CPU time of each step of the algorithm. In this experiment, and all the experiments on uniform distribution in the following, conjugate gradient method is used to solve the optimization problem.

\begin{figure}
\centering
\includegraphics[width=3.2in]{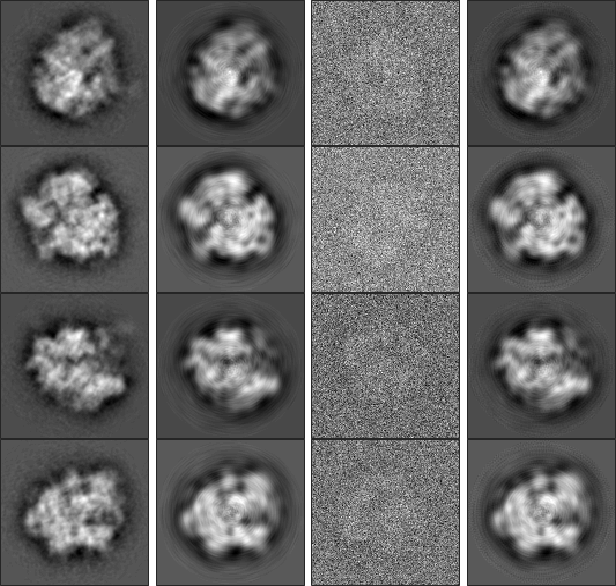}
\caption{\textbf{First column} \rev{(left to right)}: Groundtruth images before sPCA. \textbf{Second column: }Groundtruth images after sPCA.  \textbf{Third column: }examples of noisy observations. \textbf{Fourth column: }recovered images by our algorithm, rotated and permuted to align with groundtruth images.}
\label{fig:uniform}
\end{figure}

\begin{table}
\caption{CPU time cost \rev{(in seconds)} of different steps for experiments on E. coli 70s ribosome with uniform and non-uniform distribution (Corresponding to Figure~\ref{fig:uniform} and~\ref{fig:nonuniform}).}
\label{tab:time}
\vspace{3mm}
\centering
\begin{tabular}{|p{4cm}|cc|}
\hline
\backslashbox{Step}{Distribution} & Uniform  & Non-uniform  \\
\hline
Computing sPCA &  $252.2$s & $259.7$s \\ \hline
Computing mixed invariants &  $16.3$s & $17.5$s \\ \hline
Optimization &  $428.2$s & $7157.3$s \\
\hline
\end{tabular}
\end{table}

We conducted experiments to study how the recovery error increases with the noise level. As before, we set $K=10$, $\pi$ is the (known) uniform distribution, and the number of observations is $10^4$. 
Figure~\ref{fig:noise} shows the recovery results and relative errors for different SNRs.
As can be seen, both the relative sPCA error and estimation error are, more or less, inversely proportional to the SNR. When the noise is larger, the sPCA gives less coefficients and results in larger sPCA errors. In addition, the estimation of the invariants, and thus the coefficients of groundtruth images, are less accurate under larger noise. Of course, when the noise is larger, we need more observations to average out the noise.

\begin{figure*}
\centering
\includegraphics[width=6.5in]{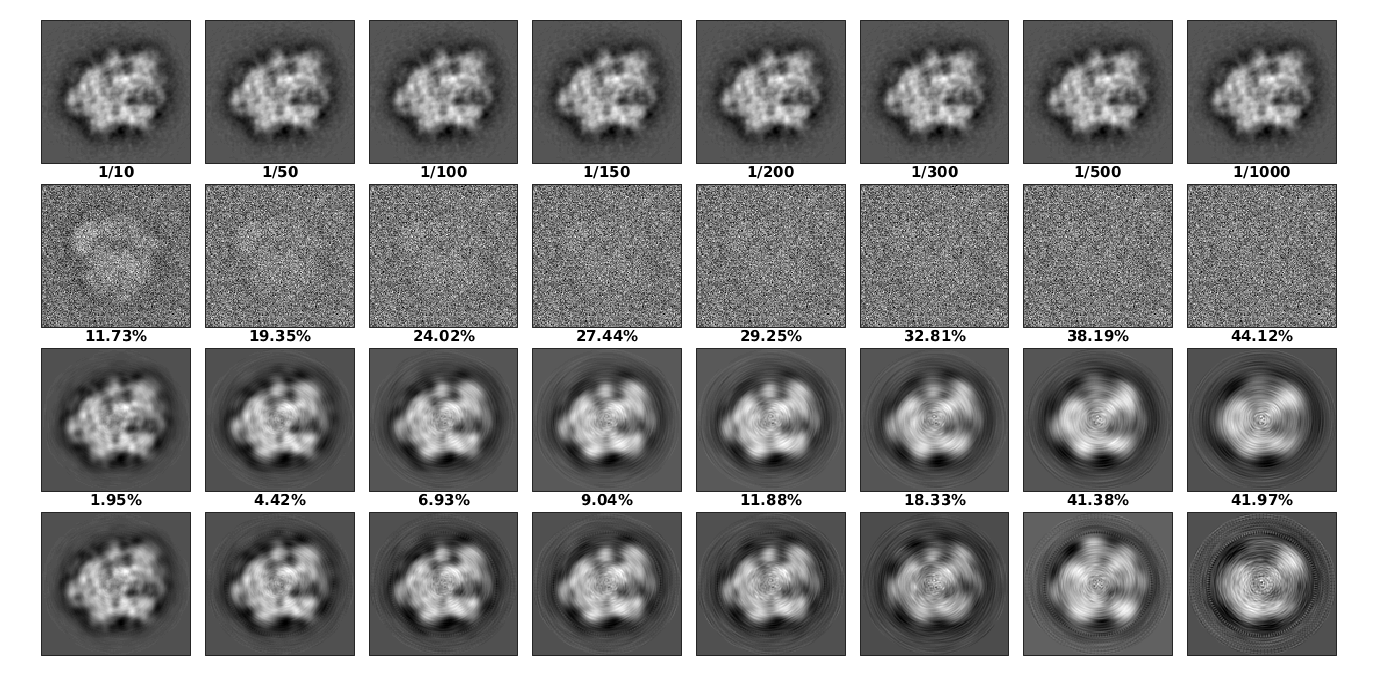}
\caption{Recovery results for different SNR. The first to fourth rows of images are groundtruth images, noisy images, groundtruth images after sPCA, and recovered images, respectively. The first to third rows of numbers are SNR, relative sPCA error compared to groundtruth images (the first row), and relative estimation error compared to groundtruth images after sPCA (the third row), respectively. From left to right, the numbers of coefficients chosen by sPCA are $172$, $88$, $60$, $45$, $41$, $33$, $20$ and $13$.}
\label{fig:noise}
\end{figure*}

The next experiment aims to examine our algorithm when optimizing over $\hat{\pi}$ and the images simultaneously. As before, we fix $K=10$ and SNR$=1/50$. We take $500$ observations for each class for the first $5$ classes, and $1500$ observations for the other $5$ classes, so that
\begin{equation}\label{eq:w}
\pi_i=0.05,\ i=1,...,5;\ \ \pi_i=0.15,\ i=6,...,10.
\end{equation}
After applying our algorithm, the TV distance between $\hpi$ and $\pi$ turns out to be $0.0086$. The relative estimation error of all $10$ images compared to groundtruth images after sPCA is $6.05\%$. The relative estimation error of the $5$ images with $\pi_i=0.05$ is $7.39\%$, while the relative estimation error of the other $5$ images with $\pi_i=0.15$ is $3.94\%$. Empirically, a non-uniform distribution does not influence much the overall quality of the recovery, though it seems underrepresented images suffer more.
Figure~\ref{fig:nonuniform} shows some results of this experiment with non-uniform distribution. CPU time cost is shown in table~\ref{tab:time}. For the experiments with non-uniform distribution, trust-regions algorithm is used to solve the optimization problem. 
Usually, trust-regions algorithm provide more accurate estimates than conjugate gradient at the cost of running times.

\begin{figure}[!t]
\centering
\includegraphics[width=3.2in, height=3.05in]{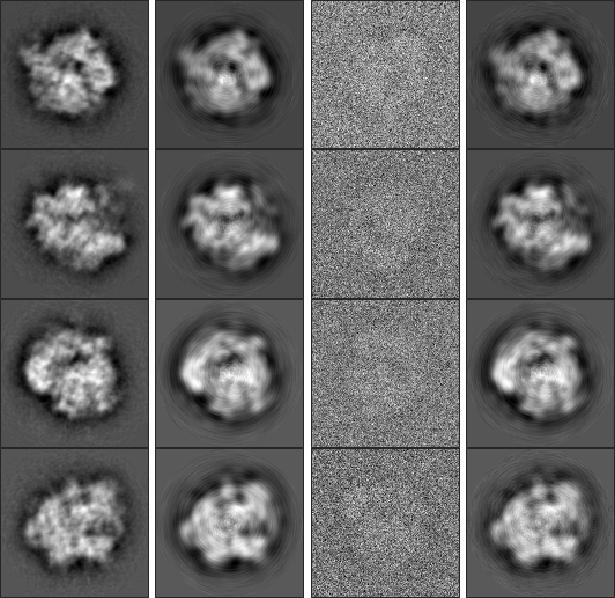}
\caption{\textbf{First column: }Groundtruth images before sPCA. \textbf{Second column: }Groundtruth images after sPCA.  \textbf{Third column: }examples of noisy observations. \textbf{Fourth column: }recovered images by our algorithm, rotated to align with groundtruth images. \textbf{The rows: }The first two rows are images with $\pi=0.15$ and the last two rows are images with $\pi_i=0.05$.}
\label{fig:nonuniform}
\end{figure}

Next, we study the influence of $\pi$ on the quality of the recovery. In this experiment, we take $K=2$, SNR$=1/50$, and $N=10^4$. Table~\ref{tab:different_pi} shows the relative estimation errors on recovered images and the TV error on the recovered distribution as the distribution $(\pi_1, \pi_2)$ shifts away from uniform. From the table, we can see that in all cases our estimated distributions are close to the true distribution with TV distance less than $0.01$. When the number of observations for different images are not equal, the image with more observations tends to have lower estimation error. 

\begin{table}
\caption{Relative estimation error compared to groundtruth images after sPCA and total variation error of experimental results with varying distribution $\pi = (\pi_1, \pi_2)$; $\mathrm{error}_i$ represents the relative error on image $i = 1, 2$.}
\label{tab:different_pi}
\centering
\begin{tabular}{|c|ccccc|}
\hline
$\pi_1$ & 0.1&0.2&0.3&0.4&0.5\\
\hline
$\dist_{\mathrm{TV}}$ & $0.0025$ & $0.0007$ & $0.0022$ & $0.0009$ & $0.0014$ \\
\hline
$\mathrm{error}_1$ & $3.59\%$ & $5.01\%$ & $1.59\%$ & $3.60\%$ & $3.95\%$ \\
\hline
$\mathrm{error}_2$ & $2.01\%$ & $4.66\%$ & $0.91\%$ & $4.08\%$ & $4.63\%$ \\
\hline
\end{tabular}
\end{table}

Usually the projection images are not perfectly centered because the particle picking is not ideal. 
We conducted experiments to examine the robustness of our algorithm to small shifts of the input images, albeit our model does not take these shifts into account. This time, we consider a problem with only one class ($K=1$) of  size $129\times129$, $5\times10^3$ noisy observations and SNR$=1/50$. A random shift is applied to each observation. The shift is generated by a 2-D uniform distribution on all the shifts within a circle of radius $s$. When $s$ ranges from $0$ to $5$, the relative estimation errors of the recovered images are $3.30\%$, $4.86\%$, $6.58\%$, $10.28\%$, $11.64\%$ and $18.30\%$, respectively.
While the error increases with the size of the shifts, it does so at a reasonable pace: when the shifts are small, the recovery errors are too.

Comparison of our algorithm with the EM method was made.
\rev{We implement a vanilla version of the EM algorithm, which is applied on sPCA coefficients rather than the images themselves, and considers only a finite set of in-plane rotations. Our EM is different from the EM-based algorithms implemented in cryo-EM software packages (such as RELION), which are more sophisticated and include many heuristics to improve running time and accuracy. Here we aim to underscore the resolution-computational load trade-off of EM}
In the first experiment, we take $K=5$ classes and $10^3$ observations per class with SNR$=1/50$.
Figure~\ref{fig:em} compares the relative estimation error compared to groundtruth images after sPCA for each class and the running time with different number of in-plane rotations considered by EM. From the figure we can see that EM performs better when the number of rotations is large ($\geq32$ in this experiments). However, at the same time EM becomes time-costing, taking nearly $10$ times more CPU-time than our algorithm. In the right plot of Figure~\ref{fig:em}, the time costs of computing invariant features are shown, which are already included in the time cost of our algorithm (the red line). Figure~\ref{fig:em2} shows the results of another experiment for nonuniform distribution and larger noise. We take the distribution~\eqref{eq:w} and SNR$=1/100$, with $K=10$ and $N=10^4$. We can observe similar phenomena as the last experiment.
\begin{figure*}
\centering
\includegraphics[width=3in]{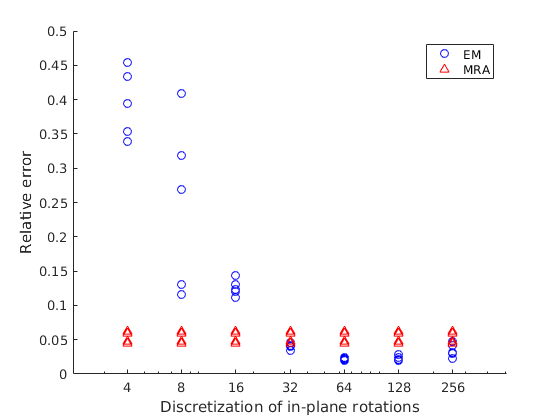}
\includegraphics[width=3in]{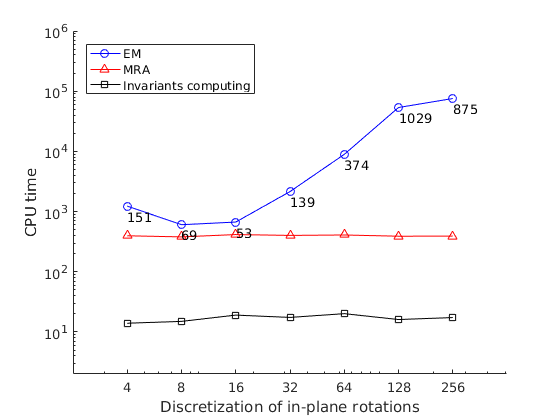}
\caption{Comparison of our algorithm (MRA) with EM in terms of relative estimation error and computation time. Images are generated from $5$ classes with uniformly random in-plane rotations. Contrary to MRA, EM needs to assume the rotations are selected from a discrete set. Here, we see the accuracy/computation time trade-off of EM for SNR$=1/50$. Left panel: for each discretization value (i.e., number of uniformly sampled angles) and algorithm, each point represents the relative estimation error compared to groundtruth images after sPCA for one of the classes. Right panel: integers indicate the number of EM iterations. The CPU time is measured in seconds.}
\label{fig:em}
\end{figure*}

\begin{figure*}
\centering
\includegraphics[width=3in]{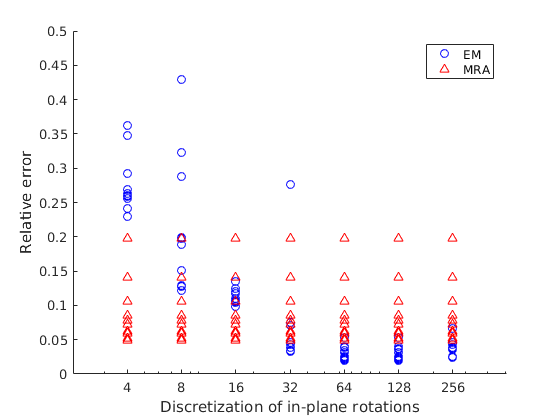}
\includegraphics[width=3in]{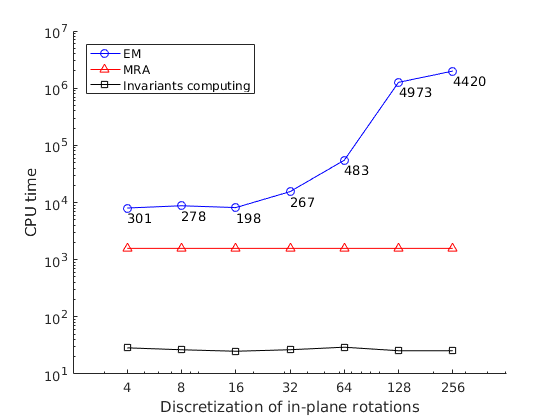}
\caption{Comparison of our algorithm (MRA) with EM in terms of relative estimation error and computation time. The 10 classes are distributed nonuniformly according to~\eqref{eq:w} and SNR$=1/100$. Left panel: for each discretization value and algorithm, each point represents the relative estimation error compared to groundtruth images after sPCA for one of the classes. Right panel: integers indicate the number of EM iterations. The CPU time is measured in seconds.}
\label{fig:em2}
\end{figure*}

To demonstrate that our algorithm applies to other data sets as well, we considered two additional molecules: the transient receptor potential cation channel subfamily V member 1 (TrpV1)~\cite{gao2016trpv1} and the yeast mitochondrial ribosome~\cite{desai2017structure}. The volumes are downloaded from The Electronic Microscopy Data Bank\footnote{\url{www.emdatabank.org}}. Similarly to the experiments before, groundtruth images are projections randomly generated from the volumes. The size of the projections is $181\times181$ and we set $K=10$, SNR = 1/50, and $N = 10^4$. Figure~\ref{fig:mole2} shows some results of TrpV1 data. In this experiment, the distribution is taken to be uniform and the variable $\hpi$ is fixed. 
The relative estimation error compared to groundtruth images after sPCA is $4.29\%$. Figure~\ref{fig:mole3} shows part of the results from the yeast mitochondrial ribosome data. In this experiment, we use the distribution given by~\eqref{eq:w}. The total variation error of $\hpi$ is $0.0141$. The relative estimation error of all the images compared to groundtruth images after sPCA is $6.03\%$, and $3.83\%$ and $7.96\%$ for images with $\pi_i=0.15$ and $\pi_i=0.05$, respectively.

\begin{figure}[t!]
\centering
\includegraphics[width=3.2in]{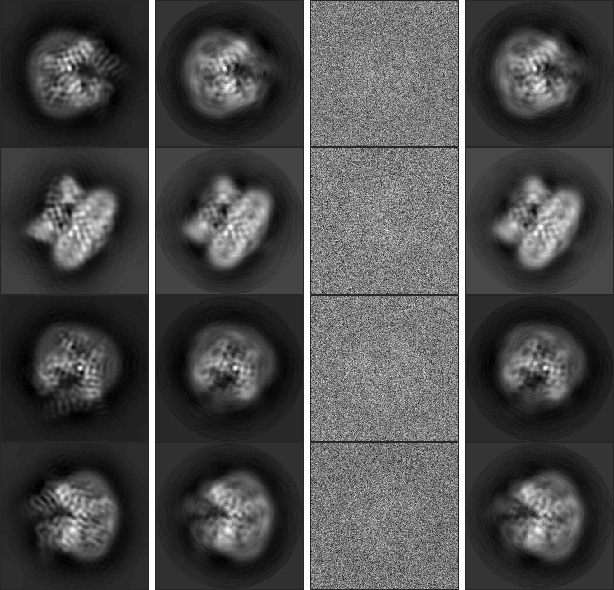}
\caption{Recovery results of the TrpV1 data. \textbf{First column: }Groundtruth images before sPCA. \textbf{Second column: }Groundtruth images after sPCA.  \textbf{Third column: }examples of noisy observations. \textbf{Fourth column: }recovered images by our algorithm, rotated to align with groundtruth images.}
\label{fig:mole2}
\end{figure}

\begin{figure}
\centering
\includegraphics[width=3.2in]{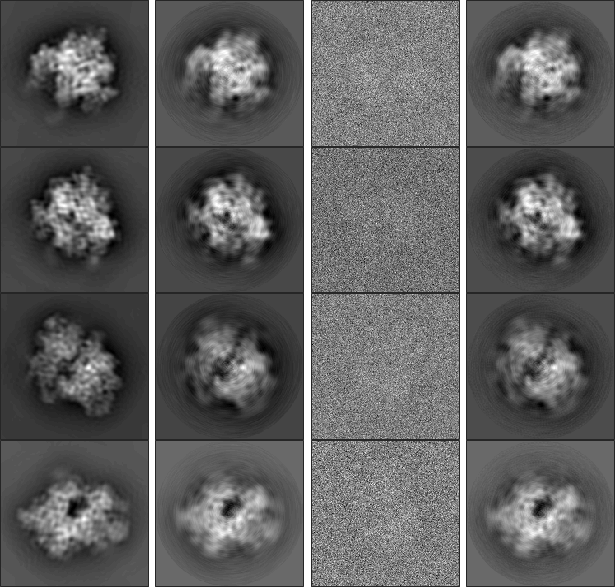}
\caption{Recovery results of the yeast mitochondrial ribosome data. \textbf{First column: }Groundtruth images before sPCA. \textbf{Second column: }Groundtruth images after sPCA.  \textbf{Third column: }examples of noisy observations. \textbf{Fourth column: }recovered images by our algorithm, rotated to align with groundtruth images. \textbf{The rows: }The first two rows are images with $\pi=0.15$. The last two rows are images with $\pi_i=0.05$.}
\label{fig:mole3}
\end{figure}

\section{Conclusion}\label{sec:con}
In this paper, we studied the problem of heterogeneous MRA for 2-D images and proposed a new algorithmic framework for 
  $2$-D classification for SPR. 
  Experimental results show that our algorithm can provide high-quality recovery of the groundtruth images (class averages), even when the noise level is high. The algorithm requires only one pass over the data and thus  suits for large experimental data sets.

In practice, the projection images in cryo-EM suffer from small random shifts. Hence, a more accurate generative model reads
\begin{equation*}
Y=T_sR_{\xi}I_\pi+\varepsilon,
\end{equation*}
where $T_s$ is a small random shift by $s$; compare with~\eqref{eq:model}. 
In future work we intend to extend our framework to take  shifts into account.
A recent study~\cite{abbe2017multireference} shows that non-uniform distributions of translations makes MRA easier. We may take this issue into account in the future. 
Meanwhile, we have shown that our algorithm is robust against small shifts.

More importantly, our algorithm considers a discrete set of viewing directions. Yet, more realistically, cryo-EM micrographs contain projections sampled from a continuous distribution of viewing directions. We hope to extend our algorithm to the continuous case in the future. 
\rev{To apply the proposed techniques to experimental data, it is necessary to handle effects of the contrast transfer functions (CTF) and of colored noise as well.}

In~\cite{boumal2018heterogeneous,weinthesis}, it was shown that the number of classes that can be demixed in 1-D hMRA is, approximately, $\sqrt{L}$, where $L$ is the length of the signals. Our  experiments indicate that we can demix $40-50$ classes. How this number depends on the size of the image or number of sPCA coefficients is left for future study.

\section*{Acknowledgment}
The research was partially supported by Award Number R01GM090200 from the NIGMS, FA9550-17-1-0291 from AFOSR, Simons Foundation Math+X Investigator Award, and the Moore Foundation Data-Driven Discovery Investigator Award.
NB is partially supported by NSF award DMS-1719558.

\bibliographystyle{ieeetr}
\bibliography{mra_bib}

\appendix
\section{Gradient of the Objective Function}\label{sec:grad}
In this section we give the gradient of our objective function \eqref{eq:obj} in Euclidean space. 
For complex variables $\hat{a}$, we treat them as a matrix, and define the gradient $\partial F/\partial \hat{a}$ as the only matrix $g=g(\hat{a},\hat{\pi})$ satisfying
\begin{equation}
\Re\{\Tr(g^*y)\} = D F(\hat{a},\hat{\pi})[y],
\end{equation}
where $y$ is a matrix with the same size as $\hat{a}$, $\Re\{\Tr(g^*y)\}$ is an inner product, and $D F(\hat{a},\hat{\pi})[y]$ is the directional derivative of $F$ at $\hat{a}$ along $y$.
As the objective function is a summation of least squares, the gradient of the objective function is the summation of the gradient of least squares terms. Hence, we only need to compute gradients for the following $3$ groups of terms:
\begin{align}
&\left|\sum\limits_{i=1}^K \hpi_i \ha^i_{0,q} - \hm_{0,q}\right|^2, \label{eq:ls1}\\
&\left|\sum\limits_{i=1}^K \hpi_i \ha^i_{k,q_1}\overline{\ha^i_{k,q_2}} - \hp_{k,q_1,q_2}\right|^2, \label{eq:ls2}\\
&\left|\sum\limits_{i=1}^K \hpi_i \ha^i_{k_1,q_2}\ha^i_{k_2,q_2}\overline{\ha^i_{k_1+k_2,q_3}} - \hb_{k_1,k_2,q_1,q_2,q_3}\right|^2.\label{eq:ls3}
\end{align}
We call \eqref{eq:ls1}, \eqref{eq:ls2} and \eqref{eq:ls3} the first-, second- and third-order terms according to the order of moments they contain. 

\subsection{First-order terms}
Let 
\begin{equation*}
\cM_q(\ha,\hpi)=\sum\limits_{i=1}^K \hpi_i \ha^i_{0,q} - \hm_{0,q}.
\end{equation*}
Then, for real variables $\hpi_i$ we can easily get
\begin{equation}
\frac{\partial |\cM_q(\ha,\hpi)|^2}{\partial \hpi_i}=2\Re\left\{\cM_q(\ha,\hpi)\ha^i_{0,q}\right\}.
\end{equation}
For complex variables $\ha^i_{0,q}$, we have
\begin{equation}
\frac{\partial |\cM_q(\ha,\hpi)|^2}{\partial \ha^i_{0,q}}=2\cM_q(\ha,\hpi)\hpi_i.
\end{equation}
For $\ha^i_{k',q'}$ with $k'\neq0$ or $q'\neq q$, we always have
\begin{equation}
\frac{\partial |\cM_q(\ha,\hpi)|^2}{\partial \ha^i_{k',q'}}=0,
\end{equation}
as they do not appear in the term. In the following subsections, we ignore the gradient with respect to such variables.

\subsection{Second-order terms}
Let
\begin{equation*}
\cP(\ha,\hpi):=\cP_{k,q_1,q_2}(\ha,\hpi)=\sum\limits_{i=1}^K \hpi_i \ha^i_{k,q_1}\overline{\ha^i_{k,q_2}} - \hp_{k,q_1,q_2}.
\end{equation*}
Then, for $\hpi_i$, similar to the last subsection, we have
\begin{equation}
\frac{\partial |\cP(\ha,\hpi)|^2}{\partial \hpi_i}=2\Re\left\{\cP(\ha,\hpi)\ha^i_{k,q_1}\overline{\ha^i_{k,q_2}}\right\}.
\end{equation}
For $\ha^i_{k,q}$, there are two cases. If $q_1=q_2$, then
\begin{equation}
\frac{\partial |\cP(\ha,\hpi)|^2}{\partial \ha^i_{k,q_1}}=4\Re\left\{\cP(\ha,\hpi)\right\}\ha^i_{k,q_1}\hpi_i.
\end{equation}
If $q_1\neq q_2$, then we have
\begin{equation}
\frac{\partial |\cP(\ha,\hpi)|^2}{\partial \ha^i_{k,q_1}}=2\cP(\ha,\hpi)\ha^i_{k,q_2}\hpi_i,
\end{equation}
and
\begin{equation}
\frac{\partial |\cP(\ha,\hpi)|^2}{\partial \ha^i_{k,q_2}}=2\overline{\cP(\ha,\hpi)}\ha^i_{k,q_1}\hpi_i,
\end{equation}

\subsection{Third-order terms}
Let 
\begin{eqnarray*}
\cB(\ha,\hpi)&:=&\cB_{k_1,k_2,q_1,q_2,q_3}(\ha,\hpi)\\
         &=&\sum\limits_{i=1}^K \hpi_i \ha^i_{k_1,q_2}\ha^i_{k_2,q_2}\overline{\ha^i_{k_1+k_2,q_3}} - \hb_{k_1,k_2,q_1,q_2,q_3}.
\end{eqnarray*}
Then, firstly we have
\begin{equation}
\frac{\partial |\cB(\ha,\hpi)|^2}{\partial \hpi_i}=2\Re\left\{\cB(\ha,\hpi)\ha^i_{k_1,q_2}\ha^i_{k_2,q_2}\overline{\ha^i_{k_1+k_2,q_3}}\right\}.
\end{equation}
Next, again we consider two cases. If $k_1=k_2$ and $q_1=q_2$, then
\begin{eqnarray}
\frac{\partial |\cB(\ha,\hpi)|^2}{\partial \ha^i_{k_1,q_1}}&=&4\cB(\ha,\hpi)\overline{\ha^i_{k_1,q_1}}\ha^i_{k_1+k_2,q_3}\hpi_i, \\
\frac{\partial |\cB(\ha,\hpi)|^2}{\partial \ha^i_{k_1+k_2,q_3}}&=&2\overline{\cB(\ha,\hpi)}(\ha^i_{k_1,q_1})^2\hpi_i.
\end{eqnarray}
Otherwise, we have
\begin{eqnarray}
\frac{\partial |\cB(\ha,\hpi)|^2}{\partial \ha^i_{k_1,q_1}}&=&2\cB(\ha,\hpi)\overline{\ha^i_{k_2,q_2}}\ha^i_{k_1+k_2,q_3}\hpi_i, \\
\frac{\partial |\cB(\ha,\hpi)|^2}{\partial \ha^i_{k_2,q_2}}&=&2\cB(\ha,\hpi)\overline{\ha^i_{k_1,q_1}}\ha^i_{k_1+k_2,q_3}\hpi_i, \\
\frac{\partial |\cB(\ha,\hpi)|^2}{\partial \ha^i_{k_1+k_2,q_3}}&=&2\overline{\cB(\ha,\hpi)}\ha^i_{k_1,q_1}\ha^i_{k_2,q_2}\hpi_i.
\end{eqnarray}

\end{document}